\documentclass[11pt,a4paper]{article}

\usepackage[T1]{fontenc}
\usepackage[utf8]{inputenc}
\usepackage{mathptmx}
\usepackage[scaled=0.9]{helvet}
\usepackage{courier}
\usepackage{amsmath,amssymb,amsthm}
\usepackage{booktabs}
\usepackage{array}
\usepackage{microtype}
\usepackage{enumitem}
\usepackage{seqsplit}
\usepackage{tikz}
\usetikzlibrary{shapes.geometric,arrows.meta,positioning,calc}
\usepackage[hidelinks]{hyperref}
\usepackage[a4paper,margin=2.6cm]{geometry}

\DeclareRobustCommand{\pth}[1]{\texttt{\seqsplit{#1}}}
\newtheorem{theorem}{Theorem}
\newtheorem{proposition}{Proposition}
\newtheorem{remark}{Remark}

\setlist[enumerate]{itemsep=0.3em,topsep=0.4em}
\setlist[itemize]{itemsep=0.3em,topsep=0.4em}
\title{AlchemQ: Proof-Carrying Quantum Circuit Optimization\\
with Per-Result Equivalence Certificates}
\author{Adam Laabs\\
  \small TriStiX S.L.\\
  \small\texttt{adam.laabs@tristix.com}}
\date{}

\begin{document}
\maketitle

\begin{abstract}
We present AlchemQ v0.5, a small proof-of-concept system that couples a
beam-search optimizer with a machine-checkable per-result certification layer
and a versioned certificate protocol (0.2.0), so that every optimized circuit
ships with a machine-verifiable artifact rather than a bare claim. The certifier
proves equivalence up to global phase by ZX-calculus full reduction, with a
numeric-tensor fallback for small circuits based on the optimal
Hilbert--Schmidt overlap; the certificate format is self-contained and
tamper-evident (canonical \texttt{gate-canon-v1} hashes, measured residuals,
tri-state verdicts \texttt{certified}/\texttt{rejected}/\texttt{inconclusive},
and versioned phase-note schemas for cross-platform reproducibility). The agent
explores rewrite suggestions under three fuzzy t-norm aggregates and is
structurally incapable of returning an uncertified circuit; since v0.4 it
additionally guarantees no componentwise regression against the original
circuit, not merely aggregate non-inferiority against the baseline. On a
reproducible benchmark of 100 circuits (QASMBench, Feynman, and generated
Clifford+T instances, gates G0, G1 and G2 passed), all 400 optimizations
terminate without errors or timeouts, every returned circuit is certified
equivalent, every mutation is detected, and the 2998-test suite passes on two
platforms (Linux/Python 3.12, macOS/Python 3.14). The PyZX-based baseline is
strong: over the 82 circuits with nonzero T-count it reduces the T-count by
21.4\% on average. Recomputed under the current zero-baseline score semantics,
the agent is never worse than the baseline anywhere and is strictly better on
9/100 circuits under the G\"odel and product norms (8/100 under the semantics of
the archived run); the three t-norms return identical final circuits on all 100
standard benchmark circuits, diverging only on 4/38 and 3/38 instances of an
adversarial 38-circuit ablation suite. Two certification case studies are new: a
false negative on \texttt{variational\_n4} root-caused to a pivot-normalization
bug in PyZX's \texttt{compare\_tensors} (pivot amplitude $4.7\times10^{-9}$,
noise amplification ${\sim}\,2\times10^{8}$; the optimal-overlap residual is
$7.4\times10^{-11}$), and eight v0.4 certificates rejected on macOS whose root
cause was BLAS-dependent raw floats in \texttt{phase\_note} --- both fixed, and
all three artifact sets (400 legacy V1, 400 historical V2, 400 fresh V2
certificates, plus a determinism audit) now validate 400/400 on both platforms.
A first hardware experiment on IBM Heron~r2 (\texttt{ibm\_kingston}) shows the
practical stakes: compiled to the native gate set, the certified-optimized
benchmark circuit is 78\% shallower with 65\% fewer two-qubit gates than the
original. Point estimates favor the optimized circuit on all three
output-quality metrics, but the differences are not statistically significant at
1024 shots per arm. We position AlchemQ as a proof-carrying optimizer, not a
verified compiler, and openly release the certificate specification and a
standalone reference verifier (Apache-2.0), together with data, a release
manifest, and validation scripts; the optimization engine is proprietary (Code
and Data Availability). The intended audience is researchers and engineers
working on practical, trustworthy quantum compilation pipelines.
\end{abstract}

\section{Introduction}

Quantum circuit optimization is a correctness-critical transformation. The
optimized circuit is executed in place of the original, often on hardware where
a single wrong gate invalidates an expensive experiment. The toolchain offers
weak evidence that the transformation preserved meaning. Production compilers
apply long sequences of rewrites with no per-result guarantee, and optimization
bugs in widely used compilers are reported and fixed regularly: Qiskit issues
\#7961 and \#7167 are two examples \cite{qiskit7961,qiskit7167}. Each such bug
silently corrupted every circuit it touched until someone noticed.

Formal verification has confirmed that the concern is not hypothetical. CertiQ
detected three bugs in the Qiskit implementation while verifying 26 of its 30
compiler passes \cite{certiq}. The Giallar toolkit found three further critical
bugs, confirmed by the Qiskit team, while verifying 44 of 56 passes across 13
versions \cite{giallar}. Meanwhile the rewriting passes themselves are
increasingly heuristic or learned: template matching, ZX-calculus
simplification, reinforcement-learning agents \cite{alphatensorq,rlzx}. The risk
shifts from a slow optimizer to a silently wrong one.

A certification layer that fails loudly is what surfaces such problems. As
reported in Section~\ref{sec:eval}, AlchemQ's certificates caught a
normalization bug in PyZX's tensor-comparison routine:
\texttt{compare\_tensors} amplifies numerical noise by a factor of order
$2\times10^{8}$ when the pivot element sits near the cancellation floor. A
correct optimization had been made to look unverifiable.

AlchemQ explores a simple alternative. Optimization and verification are treated
as separate concerns, connected by a contract, and the contract is the only
trusted component. The optimizer is a fixed-seed beam search over a small
catalogue of rewrite suggestions and is in no way trusted. The certificate layer
is independent of the search, re-proves semantic equivalence from scratch for
every emitted artifact, and either succeeds loudly or rejects the candidate.
Since v0.4 the optimizer additionally guarantees that the returned circuit never
regresses any objective componentwise against the original, closing a gap where
ties at aggregate zero could previously be resolved in favour of a regressing
candidate.

The contribution is in three parts.

\begin{enumerate}
\item \textbf{An open certification protocol for circuit pairs.} Certificates
  are self-contained JSON artifacts carrying a canonical hash
  (\texttt{gate-canon-v1}), a declared equivalence method, tri-state verdicts,
  measured residuals, versioned phase-note schemas, and a 128-bit
  protocol-tagged identifier. An independent verifier recomputes the proof from
  the embedded circuits alone and returns a structured
  \texttt{VerificationResult} with a reason code (Section~\ref{sec:protocols}).
  Protocol 0.2.0 ships with a strict migration path from protocol 0.1.1; 400
  migrated and 400 freshly generated certificates validate 400/400 on two
  platforms.

\item \textbf{A reference certifier and verifier.} Equivalence is proved up to
  global phase by ZX-calculus full reduction, with a numeric tensor fallback for
  circuits of at most six qubits whose relative phase is estimated by the
  optimal Hilbert--Schmidt overlap rather than by a single pivot element.
  Certificates record how the global phase was determined (measured, from the
  witness, or explicitly not tracked) under a canonical phase representation in
  $[-\pi,\pi)$.

\item \textbf{A toy search agent with structural safety guarantees.} Under three
  standard fuzzy t-norms (G\"odel, product, {\L}ukasiewicz), continuing the
  t-norm aggregation programme initiated in the LGGT+ framework for EU AI Act
  compliance classification \cite{lggt}, the agent never returns an uncertified
  circuit, never scores below the PyZX-based baseline, and never returns a
  circuit that regresses any objective componentwise against the original.
  Empirically the three norms agree on the final circuit for all 100 standard
  benchmark circuits and diverge only on a handful of adversarial ablation
  instances, a boundary characterized exactly in
  Sections~\ref{sec:scoring}--\ref{sec:eval}.
\end{enumerate}

To the best of our knowledge, AlchemQ is the first reported quantum circuit
optimizer that emits, for every returned circuit pair, a self-contained,
versioned, tamper-evident equivalence certificate designed for offline
third-party re-verification under an openly specified artifact protocol. We do
not claim the first per-instance equivalence check; concurrent work verifies
substitutions inside a compilation pipeline \cite{bartkiewicz}. We do not claim
the first error-budget objective. The claim concerns the proof-carrying artifact
and its independent re-verification (Section~\ref{sec:related}).

This paper is a pilot study in the style of \cite{lggt}. The system is small,
the claims are bounded, and the limitations are stated in full in
Section~\ref{sec:limits}.

We are explicit about what AlchemQ is not. It is not a learned optimizer; there
is no RL component yet. It is not a competitive replacement for industrial
passes. It is not a formally verified compiler in the sense of VOQC or QBlue:
the guarantee applies to each emitted artifact rather than to the optimizer's
code, and the verifier currently shares its ZX implementation with the certifier
(Section~\ref{sec:limits}).

The rest of the paper is organized as follows. Section~\ref{sec:background}
collects background and Section~\ref{sec:related} positions the work.
Section~\ref{sec:system} describes the architecture and
Section~\ref{sec:protocols} the certification protocols.
Sections~\ref{sec:phase}--\ref{sec:trail} detail phase handling, the certificate
format, and proof trails. Section~\ref{sec:scoring} defines the scoring and the
safety guarantees. Section~\ref{sec:eval} reports the experiments, including two
certification case studies. Sections~\ref{sec:discussion}--\ref{sec:conclusion}
discuss, limit, and conclude.

\section{Background}\label{sec:background}

\paragraph{Quantum circuits and metrics.}
A quantum circuit is a sequence of gates acting on $n$ qubits \cite{nielsen}. We
target three standard cost metrics: T-count (number of $T$/$T^\dagger$ gates),
two-qubit-gate count (CNOT/CZ), and circuit depth. Clifford gates are
efficiently simulable \cite{gottesman}, so the genuinely costly resources in
fault-tolerant implementations are the non-Clifford (T) content \cite{amymosca}
and the entangling structure; these metrics are the objectives used by the
optimizer throughout the paper.

\paragraph{ZX-calculus.}
The ZX-calculus is a graphical language for quantum linear maps in which
circuits become diagrams built from Z- and X-spiders, and rewriting corresponds
to sound equational transformations \cite{coecke2011,coecke2017,wetering2020}.
It is complete for the stabilizer fragment \cite{backens} and for Clifford+T
quantum mechanics \cite{jeandel}, and underlies state-of-the-art circuit
extraction and optimization \cite{pyzx}; a recent review systematizes the
resulting optimization techniques \cite{fischbach}. We use PyZX \cite{pyzx} both
for the baseline optimizer (\texttt{full\_reduce}) and as the equivalence engine
of the certifier. The PyZX version is pinned at 0.10.4 (the release was
validated against 0.9.0--0.10.4), and the certificate pins the engine version
range explicitly.

\paragraph{Equivalence up to global phase.}
Two circuits are equivalent up to global phase, written $C_1 \sim C_2$, if their
unitaries satisfy $U_1 = e^{i\phi}U_2$ for some $\phi \in [0,2\pi)$. Global
phase is unobservable for closed evolution \cite{nielsen}, so most practical
checkers adopt this notion; AlchemQ certifies equivalence in this sense and
additionally records how the phase was determined
(Section~\ref{sec:phase}).

\paragraph{Fuzzy t-norms.}
A t-norm is an associative, commutative, monotone binary operation on $[0,1]$
with identity 1. The three archetypal continuous t-norms are G\"odel
($T_G(x,y) = \min(x,y)$), product ($T_P(x,y) = x\cdot y$), and {\L}ukasiewicz
($T_L(x,y) = \max(0, x+y-1)$) \cite{hajek,klement}. We lift them to $k$-ary
aggregation operators over per-objective scores; their code-level properties are
pinned by property tests (Section~\ref{sec:scoring}).

\section{Related Work and Landscape}\label{sec:related}

Table~\ref{tab:landscape} situates AlchemQ among four families of systems. We
emphasize that the families solve different problems: verified compilers prove
compiler passes correct once and for all; equivalence checkers verify concrete
circuit pairs on demand; RL optimizers search for better circuits;
proof-carrying code (PCC) attaches machine-checkable evidence to results.

\begin{table}[t]
\centering
\footnotesize
\setlength{\tabcolsep}{4pt}
\caption{Landscape of related systems. ``Output artifact'' is what the system
emits besides the optimized circuit or verdict.}
\label{tab:landscape}
\begin{tabular}{@{}llllc@{}}
\toprule
System & Family & Guarantee & Output artifact & Open protocol? \\
\midrule
VOQC \cite{voqc}            & verified compiler      & compiler-level proof (Coq)  & circuit       & no \\
CertiQ \cite{certiq}        & verified compiler      & compiler-level proof        & circuit       & no \\
QBlue \cite{qblue}          & verified compiler      & compiler-level proof (Rocq) & circuit       & no \\
QCEC \cite{qcec}            & equivalence checker    & DD/ZX-based engines         & verdict       & no \\
MPO checker \cite{mpo}      & equivalence checker    & tensor-network contract.    & verdict       & no \\
Feynman/PathSums \cite{amy2018} & verifier           & path-sum semantics          & verdict       & no \\
Quartz \cite{quartz}        & verified optimizer     & SMT-verified rewrites       & circuit       & no \\
AlphaTensor-Q. \cite{alphatensorq} & learned optimizer & none (heuristic)         & circuit       & no \\
RL-ZX \cite{rlzx}           & learned optimizer      & none (heuristic)            & circuit       & no \\
PCC \cite{necula}           & proof-carrying code    & per-artifact proof          & proof object  & concept \\
AlchemQ (ours)              & proof-carrying optimizer & per-artifact proof        & open certificate & yes \\
\bottomrule
\end{tabular}
\end{table}

\paragraph{Verified compilers.}
VOQC \cite{voqc} formalizes a quantum optimizer in Coq and extracts verified
passes; CertiQ \cite{certiq} and Giallar \cite{giallar} verify (parts of) the
Qiskit compiler with mostly-automated tooling, finding real bugs; QBlue
\cite{qblue} more recently formalizes in Rocq a compiler for second-quantized
(fermionic) simulation circuits, proving the full pipeline from a high-level
language to qubit circuits correct. Such systems offer the strongest guarantee,
correctness of the compiler itself, but the proof effort is substantial, the
verified fragment lags behind the heuristic state of the art, and the guarantee
covers only what the compiler does, not what an external agent might propose.
AlchemQ is not a verified compiler: it is a proof-carrying optimizer whose
certificates are about individual circuit pairs, so any optimizer, even an
unverified or learned one, can sit behind the same contract.

\paragraph{Equivalence checkers and verifiers.}
QCEC \cite{qcec} combines decision diagrams, ZX-based reasoning, and simulation
into a mature checker; the intermediary matrix-product-operator approach of
Sander et al. \cite{mpo} verifies circuits via tensor-network contraction;
Feynman/PathSums \cite{amy2018} provides a functional-verification framework
based on path-sum semantics. AlchemQ's certifier is far simpler: a ZX
full-reduce path plus a brute-force tensor path for small circuits. The novelty
is not the engine but the packaging, a canonical, hash-chained, versioned
certificate format that third parties can verify independently, together with
the structural integration with the optimizer. Cross-validation of AlchemQ
certificates against QCEC remains future work (Section~\ref{sec:limits}).

\paragraph{Per-instance verification and the integrity of verdicts.}
Closest in spirit to our per-artifact contract is the concurrent and
complementary work of Bartkiewicz and Tulewicz \cite{bartkiewicz}, which couples
an error-budget objective with per-instance verification when selecting
decompositions in Toffoli networks. The difference is the evidentiary surface:
their checks are in-pipeline assertions on reachable subspaces, and at larger
widths the phase-aware substitutions are analysis-sound (spot-checked) rather
than machine-checked, whereas AlchemQ emits a self-contained, hash-chained
certificate that a third party re-verifies offline, with a tri-state verdict
that can answer \texttt{inconclusive} instead of silently rounding to a boolean.

The motivation for insisting on artifacts rather than verdicts is empirical.
Recent independent audits report unexpected behaviours of the Qiskit compiler
identified by the SQbricks hybrid equivalence checker \cite{sqbricks}, seventeen
bugs across four widely-used quantum computing platforms, fourteen of them
confirmed or fixed \cite{qite}, a documented false ``equivalent'' verdict by
QCEC under phase shifts \cite{quokka}, and further defects in quantum software
stacks found by fuzzing and equivalence-modulo-inputs testing campaigns
\cite{fuzzing,qemi}. This is a chain-of-trust problem that only independently
checkable evidence addresses. We also note that the term \emph{proof-carrying}
has very recently appeared in the quantum literature for resource-estimation
arithmetic blocks \cite{pcecdlp}; that work certifies arithmetic subcircuits and
resource counts, not optimizer output equivalence, and is orthogonal to ours.

\paragraph{Multi-objective and noise-aware compilation.}
Multi-objective treatment of compilation exists, predominantly through
NSGA-II-style evolutionary search: pass selection over QIR \cite{swierkowska},
fidelity--cost trade-offs under simulated backend calibration
\cite{rindell,ghlib}, fuzzy-steered genetic search \cite{jha}, and, most
recently, Pareto filtering with calibration-derived error proxies and Bayesian
ordering \cite{qbalance}. None of these attaches equivalence evidence to the
returned circuits, none aggregates objectives by t-norms, and T-count as an
FTQC-facing criterion has, to our knowledge, not been used as a multi-objective
criterion at all.

Separately, a noise-adaptive compilation line estimates circuit fidelity from
calibration data: weighted log-infidelity allocation \cite{murali},
fidelity-ranked layout selection \cite{nation}, $T_2$-aware mapping
\cite{tram}, log-infidelity routing \cite{finesse}, and diagnostic error
attribution \cite{hbr}. A recent reassessment, however, finds noise-adaptive
mapping barely distinguishable from random mapping and locates the practical
benefit in variance reduction rather than in the mean \cite{revisiting}. AlchemQ
v0.5 is calibration-agnostic; the planned hardware-aware layer
(Section~\ref{sec:discussion}) therefore treats calibration as a risk-adjusted,
uncertainty-bearing signal under certified equivalence, a combination we have
not found in the literature.

\paragraph{Learned optimizers.}
AlphaTensor-Quantum \cite{alphatensorq}, RL-ZX \cite{rlzx}, and large search
frameworks such as OpenEvolve \cite{openevolve} demonstrate strong empirical
gains but ship no correctness evidence; indeed, we adopt OpenEvolve-style
evaluation pipelines as inspiration for future RL integration. AlchemQ is
conservative by design: the toy agent exists to exercise the certificate
contract, and any RL policy would inherit the same hard gate.

\paragraph{Benchmarks.}
QASMBench \cite{qasmbench} and the Feynman repository \cite{amy2018,feynman}
supply the standard circuits; QUBIKOS provides layout-synthesis benchmarks with
provably known optimal SWAP costs \cite{qubikos}, a complementary idea
(known-optimal instances for routing) to our known-equivalent instances for
rewriting. Our evaluation adds generated Clifford+T instances and, crucially,
equivalence witnesses and mutation tests for every instance.

\section{System Overview}\label{sec:system}

AlchemQ v0.5 consists of seven modules (Figure~\ref{fig:pipeline}):
\texttt{tnorms.py} (t-norm definitions and properties), \texttt{scoring.py}
(per-objective scores and aggregation), \texttt{optimizer.py} (beam search with
a certified hard gate, persistent best tracking, non-regression rank, and final
re-check), \texttt{equivalence.py} (shared numeric semantics: optimal-overlap
phase estimation), \texttt{certifier.py} (certificate generation and
verification, protocol 0.1.1, with the v0.5 backported overlap semantics),
\texttt{certifier\_v2.py} (the protocol 0.2.0 layer: tri-state verdicts,
canonical hashes, versioned dispatchers, structured verification results, and
the V1$\rightarrow$V2 migration path), and \texttt{prooftrail.py} (rewrite logs
bound into certificates).

The search loop certifies with the protocol-0.1.1 engine. Protocol 0.2.0 is the
certificate layer used for release artifacts, both migration and fresh
generation, and both protocols share the same equivalence semantics through
\texttt{equivalence.py}. Total core code is approximately 3,300 lines of Python
with three runtime dependencies (\texttt{pyzx}, \texttt{numpy},
\texttt{qiskit}); the test suite comprises 2998 tests passing in approximately
80\,s on both supported platforms.

\begin{figure}[t]
\centering
\begin{tikzpicture}[
  node distance=7mm,
  box/.style={draw, rounded corners=2pt, align=center, inner sep=4pt,
              minimum width=62mm, font=\small},
  lab/.style={font=\scriptsize\itshape},
  >={Stealth[length=2mm]}
]
\node[box] (in)   {input QASM circuit};
\node[box, below=of in] (opt)
  {\textbf{untrusted optimizer}\\[-1pt]
   {\scriptsize rewrite suggestions $\times$ 3 t-norms, beam search}};
\node[box, below=of opt] (cand) {candidate rewrites};
\node[box, below=of cand] (gate)
  {\textbf{certification hard gate}\\[-1pt]
   {\scriptsize ZX full-reduce, tensor-overlap fallback}};
\node[box, below=of gate] (nr)
  {\textbf{non-regression check}\\[-1pt]
   {\scriptsize against the original circuit}};
\node[box, below=of nr] (cert)
  {certificate $+$ optimized circuit\\[-1pt]
   {\scriptsize protocol 0.2.0 artifact}};
\node[box, below=of cert] (ver)
  {\textbf{independent verifier}\\[-1pt]
   {\scriptsize offline third-party re-check}};
\node[box, right=18mm of gate, minimum width=26mm] (disc) {discard};

\draw[->] (in)   -- (opt);
\draw[->] (opt)  -- (cand);
\draw[->] (cand) -- (gate);
\draw[->] (gate) -- node[lab, right] {certified} (nr);
\draw[->] (gate) -- node[lab, above, align=center]
     {rejected /\\ inconclusive} (disc);
\draw[->] (nr)   -- (cert);
\draw[->] (cert) -- (ver);
\end{tikzpicture}
\caption{The AlchemQ proof-carrying pipeline. Every candidate rewrite must pass
the certification hard gate; rejected or inconclusive candidates are discarded,
and only certified, non-regressing circuits are returned, each with a
self-contained certificate that a third party can re-verify offline.}
\label{fig:pipeline}
\end{figure}
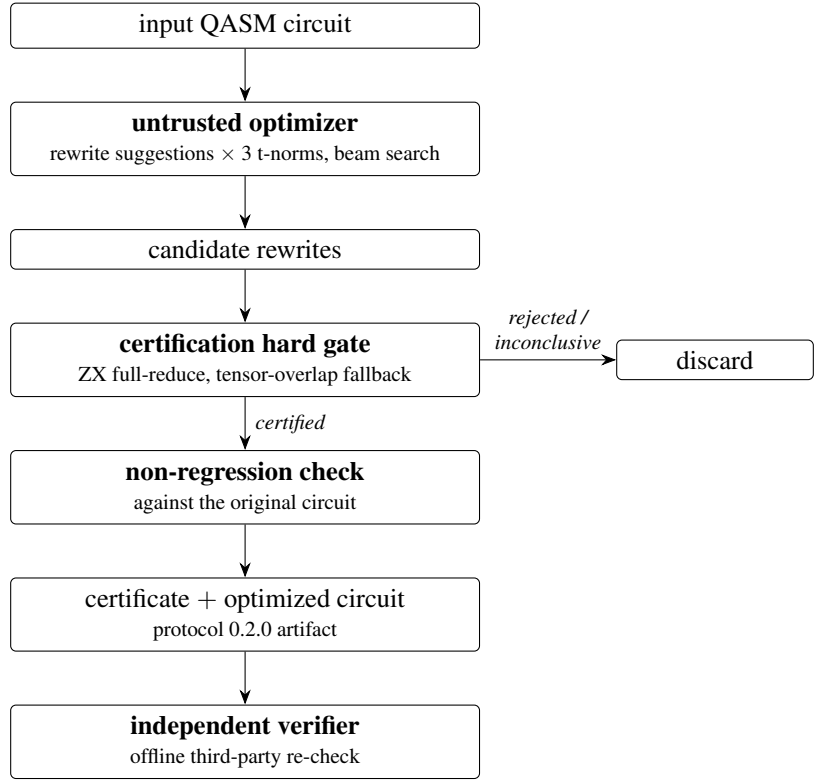

The optimizer's rewrite catalogue is small and local by design: gate-merging
windows, commutation pushes, inverse cancellation, and a PyZX-based extraction
pass (\texttt{full\_reduce} followed by extraction). Each suggestion is a
deterministic function of the current circuit and a seeded RNG. The baseline
used throughout is the same PyZX extraction pass applied once to the original
circuit.

\section{Certification Protocols}\label{sec:protocols}

AlchemQ ships two certificate protocol versions that share one equivalence
semantics. Protocol 0.1.1 (module \texttt{certifier.py}) is the legacy format
used inside the search loop; protocol 0.2.0 (module \texttt{certifier\_v2.py},
new in this version) is the release format with canonical hashing, tri-state
verdicts, and versioned note schemas. A strict migration path
(\texttt{migrate\_cert\_v1\_to\_v2}) upgrades legacy certificates, and
\texttt{verify\_any\_certificate} dispatches on the protocol tag.

\subsection{Equivalence engine (shared semantics)}\label{sec:engine}

Both protocols decide equivalence of two circuits $C_0$, $C_1$ (given as
OpenQASM 2.0 strings, parsed by PyZX) by the following pipeline:

\begin{enumerate}
\item \textbf{Qubit-count gate.} If the circuits act on different numbers of
  qubits, the verdict is \texttt{rejected}.

\item \textbf{ZX path.} Convert both circuits to ZX diagrams
  (\texttt{to\_graph}) and interleave \texttt{full\_reduce} with
  \texttt{normalize}; if the resulting diagram is a strict identity diagram (no
  vertices other than boundary inputs/outputs, no edges), the circuits are equal
  exactly, and the graph scalar gives the relative global phase. Verdict:
  \texttt{certified} with method \texttt{zx-full-reduce}.

\item \textbf{Numeric fallback ($n \leq 6$ qubits).} Compute the dense tensors
  $t_0$, $t_1$ (\texttt{to\_tensor()}) and estimate the relative phase by the
  optimal Hilbert--Schmidt overlap
  \[
    \phi^{*} \;=\; \arg\min_{\phi}\, \bigl\| t_0 - e^{i\phi} t_1 \bigr\|_F
            \;=\; \arg\bigl(\langle t_0, t_1\rangle\bigr),
    \qquad
    \langle t_0, t_1\rangle = \mathrm{Tr}\bigl(t_0^{\dagger} t_1\bigr).
  \]
  The decision is made on the phase-aligned tensors, never on a single pivot
  element (Section~\ref{sec:eval} shows why this matters): protocol 0.1.1 uses
  \texttt{np.allclose} with $\mathrm{atol} = \mathrm{rtol} = 10^{-9}$; protocol
  0.2.0 uses the normalized residual
  \[
    r \;=\; \frac{\bigl\| t_0 - e^{i\phi^{*}} t_1 \bigr\|_F}
                 {\max\bigl(\|t_0\|_F, \|t_1\|_F\bigr)}
      \;\leq\; \mathrm{atol} + \mathrm{rtol} = 2\times10^{-6},
  \]
  which is recorded as \texttt{measured\_residual} in the certificate.
  Degenerate tensors (NaN/Inf, zero, shape mismatch, or orthogonal with
  $|\langle t_0,t_1\rangle| < 10^{-15}$) are never certified: V2's engine
  rejects them outright, V1's engine reports \texttt{inconclusive}. Verdict:
  \texttt{certified} with method \texttt{numeric-tensor}, or \texttt{rejected}.

\item \textbf{Fail-soft.} ZX-inconclusive with $n > 6$ qubits, exceptions in the
  numeric path, unsupported gate classes (non-unitary records such as
  \texttt{measure} or \texttt{if}), and unstable round-trip re-serializations
  all yield \texttt{inconclusive} rather than an error or a wrong answer.
  Certificates carry the tri-state status
  $\in \{\texttt{certified}, \texttt{rejected}, \texttt{inconclusive}\}$ with
  $\texttt{certified} \Leftrightarrow \texttt{status} = \texttt{certified}$; an
  \texttt{inconclusive} certificate is evidence of nothing, positively or
  negatively.
\end{enumerate}

\subsection{Protocol 0.2.0 additions}\label{sec:v2}

Protocol 0.2.0 makes the artifact independently verifiable and byte-reproducible.

\begin{itemize}
\item \textbf{Canonical hashes.} Each embedded circuit contributes two hashes:
  \texttt{artifact\_byte\_hash} (SHA-256 of the raw QASM text, detecting any
  tampering with the embedded artifact) and \texttt{canonical\_hash} (SHA-256 of
  the canonical gate-level rendering \texttt{gate-canon-v1}: gates are
  serialized one per line in circuit execution order, since gate order is
  semantic nothing is reordered or sorted, and each gate is rendered through an
  explicit, frozen per-gate-class table, with phases represented as exact
  $\pi$-fractions; float angles are snapped via
  \texttt{Fraction.limit\_denominator($2^{20}$)}). The canonical hash detects
  changes to the circuit's mathematical content even under reformatting.
  \texttt{rewrite\_log\_hash} binds the proof trail into the certificate
  (Section~\ref{sec:trail}).

\item \textbf{Versioned dispatchers.} The phase note is governed by a versioned
  schema registry (\texttt{phase-note-v1}, \texttt{phase-note-v2}); a
  \texttt{CANONICALIZERS} registry maps \texttt{canonicalizer\_version} tags to
  rendering functions. Unknown schema or canonicalizer tags yield
  \texttt{inconclusive}, never silent acceptance. \texttt{phase-note-v1} notes
  are validated semantically (grammar check plus canonical phase equality plus
  residual consistency), which is what makes historical certificates
  reproducible across BLAS/LAPACK variants; \texttt{phase-note-v2} notes are
  rendered canonically and compare byte-exact.

\item \textbf{Structured verdicts.} Verification returns a
  \texttt{VerificationResult} with an outcome, a pipeline stage, and a
  machine-readable \texttt{reason\_code}, instead of a bare boolean.

\item \textbf{Versioned identity.} \texttt{cert\_id} is a 128-bit
  protocol-tagged hash; \texttt{supersedes} links a migrated certificate to its
  legacy parent, whose 16-hex-digit id is preserved in
  \texttt{extra.migrated\_from}.
\end{itemize}

\subsection{Assumptions and soundness}

\begin{theorem}[Soundness of certification, v0.5 semantics]
Under assumptions (A1)--(A5) below, if the certifier returns status
\texttt{certified} for a pair $(C_0, C_1)$, then
\begin{enumerate}[label=(\roman*)]
\item if the recorded method is \texttt{zx-full-reduce}, there exists $\phi$
  such that $[\![C_1]\!] = e^{i\phi}[\![C_0]\!]$ exactly, and whenever the phase
  status is \texttt{measured}, the recorded canonical phase (angle in
  $[-\pi,\pi)$, \texttt{pi\_fraction} in $[-1,1)$) equals $\phi$ (respectively
  $\phi/\pi$) modulo $2\pi$ within the phase-equality tolerance ($10^{-6}$ in
  V2; V1 additionally pins the field comparison at $10^{-9}$);
\item if the recorded method is \texttt{numeric-tensor}, the same conclusion
  holds up to the pinned numeric tolerance: in V1 the phase-aligned tensors
  agree elementwise within $\mathrm{atol} = \mathrm{rtol} = 10^{-9}$; in V2 the
  normalized Hilbert--Schmidt residual satisfies $r \leq 2\times10^{-6}$;
\item statuses \texttt{rejected} and \texttt{inconclusive} never yield
  $\texttt{certified} = \texttt{true}$, and acceptance is always decided by the
  verifier's recomputed verdict, not by the recorded fields.
\end{enumerate}
\end{theorem}

\begin{proof}[Proof sketch]
For (i), the ZX rewriting implemented by PyZX is sound: every rewrite preserves
the linear map denoted by the diagram, so if $G(C_0) - G(C_1)$ reduces to the
empty identity diagram, the denotations are equal up to the recorded scalar. For
(ii), $\phi^{*} = \arg\langle t_0, t_1\rangle$ is the minimizer of the Frobenius
distance among all global phases, and the residual (or \texttt{allclose})
criterion bounds the remaining distance; degenerate inputs are excluded by the
explicit guards. For (iii), the tri-state status is computed by the same
pipeline on both the generation and the verification side, and only the
\texttt{certified} branch sets the boolean flag. Assumptions (A1)--(A5) are
enumerated below.
\end{proof}

\begin{description}
\item[(A1) QASM parsing.] PyZX's OpenQASM 2.0 importer correctly implements the
  gate semantics for the supported gate set.
\item[(A2) ZX soundness.] PyZX's \texttt{full\_reduce} / \texttt{normalize}
  rewrites are semantics-preserving.
\item[(A3) Tensor semantics.] \texttt{to\_tensor()} computes the denotational
  tensor of the circuit.
\item[(A4) Numeric criterion and tolerances.] The numeric path estimates the
  relative phase by the optimal Hilbert--Schmidt overlap (never a single pivot
  element) and decides with pinned tolerances (V1:
  $\mathrm{atol} = \mathrm{rtol} = 10^{-9}$; V2: residual
  $\leq 2\times10^{-6}$). The conclusions of (ii) therefore hold up to the
  pinned tolerance, not exactly; for continuous gate sets the $2\times10^{-6}$
  threshold is a sharp decision boundary with a documented grey zone
  (Section~\ref{sec:limits}), while for Clifford+T circuits the observed margin
  is about four orders of magnitude on both sides (equivalent pairs
  ${\sim}\,2\times10^{-10}$, inequivalent pairs $\gtrsim 10^{-3}$).
\item[(A5) Global-phase honesty and canonicalization.] If either input declares
  a nonzero \pth{global\_phase} (which PyZX 0.10.4 drops in
  \texttt{to\_graph()}/\texttt{to\_tensor()}), the certificate must report
  \texttt{not\_tracked}. Measured phases are stored in a single canonical
  representative (angle in $[-\pi,\pi)$, \texttt{pi\_fraction} in $[-1,1)$,
  9-decimal rounding, $+\pi$ snapped to $-\pi$) and compared modulo $2\pi$.
\end{description}

\begin{remark}[Which tolerances exist]
V1 pins a single $\texttt{numeric\_tolerance} = 10^{-9}$ that serves both as the
field-comparison tolerance and, since the v0.5 backport, as the numeric-path
decision tolerance. V2 separates concerns: decision threshold
$\mathrm{atol} + \mathrm{rtol} = 2\times10^{-6}$, residual reproduction
tolerance $\texttt{RESIDUAL\_EQ\_TOL} = 10^{-6}$, phase equality
$\texttt{PHASE\_EQ\_TOL} = 10^{-6}$, and metric-field tolerance $10^{-9}$. All
are recorded in the certificate or pinned by the protocol.
\end{remark}

\begin{remark}[Why not an external SMT or DD checker?]
The ZX engine was chosen for uniformity with the baseline optimizer and for its
strong empirical coverage on Clifford+T circuits (Section~\ref{sec:eval}). The
certificate format is engine-agnostic: a future QCEC- or MPO-based engine can
emit the same artifact with a different method tag.
\end{remark}

\section{Global-Phase Honesty}\label{sec:phase}

Certificates record how the relative global phase was handled, in a structured
\texttt{global\_phase} object: status $\in$ \{\texttt{measured},
\texttt{from\_witness}, \texttt{not\_tracked}\}, a canonical angle
(\texttt{angle\_rad}) in $[-\pi,\pi)$, and a canonical \texttt{pi\_fraction} in
$[-1,1)$. The ZX path measures the phase from the surviving graph scalar; the
numeric path measures it as $\arg\langle t_0,t_1\rangle$.
\texttt{from\_witness} is used when the phase is taken from a supplied witness
rather than recomputed. \texttt{not\_tracked} is forced whenever an input
circuit declares a nonzero \pth{global\_phase}, because PyZX 0.10.4 silently
drops that declaration during graph and tensor conversion. Certifying such a
pair while claiming a measured phase would be unsound, so the certificate
discloses the limitation instead.

Canonicalization matters for reproducibility: equivalent phases modulo $2\pi$
(for example $+\pi$ and $-\pi$) share one representative, and the canonical
phase note (\texttt{phase-note-v2}) renders it byte-exactly. The legacy v0.2
corpus contained six certificates whose recorded phase note had been rendered
from a raw float that wrapped around $2\pi$; the canonicalization plus the
\texttt{phase-note-v1} semantic validator makes all of them verifiable again
(Section~\ref{sec:eval}).

\section{Certificate Format}\label{sec:format}

A protocol-0.2.0 certificate is a JSON object with a declared trust boundary.
Seventeen protected fields (\texttt{cert\_id}, both QASM artifacts, both
\texttt{artifact\_byte\_hash}es, both \texttt{canonical\_hash}es,
\texttt{rewrite\_log\_hash}, both metric blocks, \texttt{improvement\_pct},
\texttt{method}, \texttt{status}, \texttt{certified}, \texttt{phase\_note}, the
\texttt{global\_phase} object, and \texttt{measured\_residual}) are covered by
the hash and consistency checks, and any tampering with them changes the
verdict. Informational fields (environment and version metadata, timestamps,
\texttt{supersedes}, \texttt{extra}, the human-readable \texttt{circuit\_name},
and the boundary declaration itself) may differ across platforms without
affecting the verdict.

The version pins are enforced by the dispatcher rather than by the boundary: a
certificate declaring an unknown \texttt{protocol\_version},
\texttt{hash\_algo}, canonicalizer tag, or \texttt{phase\_note\_schema} yields
\texttt{inconclusive}, never silent acceptance. Parsing policy: duplicate keys
are rejected on load where the platform allows, and otherwise the last
occurrence wins.

Listing~\ref{lst:cert} shows a fresh v0.5 certificate (abbreviated; hashes
truncated) for the \texttt{simon\_n6} transpiled instance, generated by the
agent under the G\"odel norm.

\begin{figure}[p]
\small
\begin{verbatim}
{
  "protocol_version": "0.2.0",
  "cert_id": "ced39c10160c1f05dc503368aeb43288",
  "supersedes": "17095eb98465a141",
  "status": "CERTIFIED",
  "certified": true,
  "circuit_name": "qasmbench__simon_n6__transpiled",
  "method": "zx-full-reduce",
  "qasm_original":  "OPENQASM 2.0; ... (59 lines)",
  "qasm_optimized": "OPENQASM 2.0; ... (18 lines)",
  "metrics_before": {"t_count": 14, "two_qubit_count": 14,
                     "depth": 33, "gate_count": 56, "qubits": 6},
  "metrics_after":  {"t_count": 0,  "two_qubit_count": 4,
                     "depth": 7,  "gate_count": 15, "qubits": 6},
  "improvement_pct": {"t_count": 100.0, "two_qubit_count": 71.43,
                      "depth": 78.79, "gate_count": 73.21},
  "global_phase": {"status": "measured", "angle_rad": 0.0,
                   "pi_fraction": 0.0},
  "phase_note": "identity up to trivial global phase (0)",
  "phase_note_schema": "phase-note-v2",
  "measured_residual": 0.0,
  "artifact_byte_hash_original":  "9623fcf4...(64 hex)",
  "artifact_byte_hash_optimized": "f8273f20...(64 hex)",
  "canonical_hash_original":  "1a3c6d03...(64 hex)",
  "canonical_hash_optimized": "472dd54c...(64 hex)",
  "rewrite_log_hash": "6f7aeb2e...(64 hex)",
  "hash_algo": "sha256",
  "canonicalizer_name": "alchemq-canonical-circuit-form",
  "canonicalizer_version": "gate-canon-v1",
  "atol": 1e-06, "rtol": 1e-06,
  "pyzx_version": "0.10.4", "numpy_version": "2.2.5",
  "python_version": "3.12.12",
  "timestamp": "2026-07-23T14:01:55Z",
  "generator_version": "alchemq-certifier/0.5.0",
  "verifier_version": "alchemq-verifier/0.5.0",
  "extra": {"migrated_from": {"cert_id": "17095eb98465a141",
                              "protocol_version": "0.1.1",
                              "legacy_method": "zx-full-reduce",
                              "legacy_certified": true}}
}
\end{verbatim}
\caption{Fresh protocol-0.2.0 certificate
(\texttt{qasmbench\_\_simon\_n6\_\_transpiled}, agent-G\"odel,
\texttt{evidence/v05/certs/}, abbreviated). The certificate is self-contained:
the verifier needs nothing beyond this file.}
\label{lst:cert}
\end{figure}

Verification of a certificate performs: schema and trust-boundary checks;
byte-hash and canonical-hash recomputation of both artifacts; re-parsing of the
embedded QASM; re-execution of the declared equivalence method with the pinned
tolerances; phase-note validation under the declared schema version; metric
recomputation against the declared improvement percentages; and residual
reproduction within \texttt{RESIDUAL\_EQ\_TOL}. Every failure maps to a stage
and a \texttt{reason\_code} in the \texttt{VerificationResult}.

\section{Proof Trail}\label{sec:trail}

Beyond the yes/no verdict, certificates bind a proof trail: a rewrite log (one
entry per applied rewrite: rule name, matched location, before/after gate
windows) hashed into \texttt{rewrite\_log\_hash}, plus the equivalence witness.
The trail is advisory and verification never trusts it, but it makes audits and
failure diagnosis fast: when a certificate fails, the log localizes the
offending rewrite. Proof trails attach to both protocols (for V2 artifacts the
attachment recomputes the log hash), and the v0.5 corpus ships with trails for
all agent-produced certificates.

\section{Scoring, Aggregation, and Safety Guarantees}\label{sec:scoring}

\paragraph{Per-objective scores.}
Let $v_i \geq 0$ be the value of objective $i$ after optimization and
$b_i \geq 0$ its value before (T-count, two-qubit count, depth). The clamped
relative improvement is
\begin{equation}\label{eq:score}
s_i(v_i, b_i) =
\begin{cases}
\mathrm{clamp}\!\left(1 - \dfrac{v_i}{b_i},\, 0,\, 1\right), & b_i > 0,\\[2.2ex]
1, & b_i = 0 \text{ and } v_i = 0 \quad(\text{neutral}),\\[1ex]
0, & b_i = 0 \text{ and } v_i > 0 \quad(\text{regression from zero}).
\end{cases}
\end{equation}

The zero-baseline semantics is the v0.4 revision: a zero-baseline objective that
stays zero (for example the T-count of a Clifford circuit) is neutral (score 1)
rather than uniformly zero, so Clifford circuits are no longer invisible to the
aggregate; any regression from a zero baseline still scores 0. Degenerate
negative inputs score 0.

\paragraph{Aggregation.}
For a t-norm $T$ lifted to $k$ arguments, the aggregate is
$A_T = T(s_1, s_2, s_3)$; the system instantiates $T_G$, $T_P$, $T_L$ as given
in Section~\ref{sec:background}. This mirrors the use of t-norms as logical
operators of Logic Tensor Networks \cite{ltn} and the compliance aggregation of
LGGT+ \cite{lggt}. Property tests pin commutativity, associativity,
monotonicity, boundary behaviour, idempotency of $T_G$, nilpotency of $T_L$, and
the ordering $T_L \leq T_P \leq T_G$ on the implemented operators.

\paragraph{Hard gate.}
Every candidate circuit, at every search step, is certified before it may enter
the beam. A candidate whose certificate is not \texttt{certified} receives
aggregate exactly 0 and is discarded; the beam is initialized with the original
circuit (aggregate 0) and the baseline circuit, so the search always has at
least two certified candidates.

\paragraph{Rank and non-regression.}
Since v0.4, candidates are compared by the lexicographic rank
\[
  \Bigl(\text{aggregate},\; \neg\,\text{regresses},\; \textstyle\sum_i s_i,\;
        -\text{steps},\; \text{gate-count},\; \text{canonical form}\Bigr),
\]
where \emph{regresses} is the componentwise predicate $v_i > b_i$ for some
objective $i$, computed against the original circuit rather than the baseline.
The search tracks the best-by-rank candidate persistently, outside the beam so
that beam truncation cannot lose it, and a final componentwise re-check before
returning replaces the winner by the original circuit if any objective
regressed.

\begin{proposition}[Structural guarantees, strengthened]
Let the optimizer return circuit $C^{*}$ with aggregate $a^{*}$ for input
$C_0$, and let $a_b$ be the baseline's aggregate. Then:
\begin{enumerate}[label=(\roman*)]
\item \textbf{No false positives.} $C^{*}$ carries a certificate with status
  \texttt{certified} (in particular $C^{*} \sim C_0$ up to global phase, in the
  sense of Theorem~1).
\item \textbf{Baseline dominance.} $a^{*} \geq a_b$.
\item \textbf{Componentwise non-regression (since v0.4).} No objective of
  $C^{*}$ is worse than in the original: $v_i(C^{*}) \leq v_i(C_0)$ for all $i$.
\end{enumerate}
\end{proposition}

\begin{proof}[Proof sketch]
(i) The hard gate is applied to every candidate before beam insertion; the
search can only ever contain or return certified circuits, and the final
certificate is re-issued independently of the cached one, a failure there
raising loudly rather than returning an uncertified circuit. (ii) The baseline
circuit is injected into the initial beam; the persistent best-by-rank tracker
contains it from step 0 and can only be replaced by candidates of strictly
higher rank, whose first component is the aggregate. Hence the returned
aggregate is at least the baseline's. (iii) The rank prefers non-regressing
candidates at equal aggregate, and the final re-check substitutes the original,
which trivially does not regress, whenever the winner regresses. The guarantee
does not follow from beam elitism alone: the beam is truncated by size and could
evict the origin, which is why best tracking is persistent and the re-check is
explicit. Note also that (iii) is stated against the original, not the baseline:
the baseline itself regresses some objective on 74/100 benchmark circuits, since
it trades objectives freely, so componentwise dominance over the baseline would
be vacuous.
\end{proof}

\begin{remark}[Regressing candidates are invisible to the aggregate]
If a candidate regresses any objective against the original, the clamp in
\eqref{eq:score} forces that score to 0, hence
$A_{T_G} = A_{T_P} = A_{T_L} = 0$: under all three norms a regressing candidate
is indistinguishable from the identity fallback. This is why, in the archived
(pre-v0.4) benchmark run, every one of the 42/100 componentwise-regressing agent
outputs sits at aggregate 0 (Section~\ref{sec:eval}); they were ties with the
origin resolved by the old tie-break. The v0.4 rank resolves those ties in
favour of non-regressing candidates without changing any aggregate, which is
what the standard-suite re-run confirms.
\end{remark}

\paragraph{Strong equivalence versus bottleneck.}
The hard gate enforces strong equivalence, unitary equality up to global phase.
The ZX path is the bottleneck in both coverage, since it may be inconclusive on
large non-Clifford circuits, and runtime (Section~\ref{sec:eval}); the tensor
fallback covers $\leq 6$ qubits. We accept the bottleneck deliberately: a
certifier that never produces false positives but is occasionally inconclusive
is the right bias for a proof-carrying system.

\section{Evaluation}\label{sec:eval}

We evaluate four claims: (1) the certification layer is sound in practice, with
no false positives, mutations caught, and all artifact sets validating; (2) the
agent satisfies its guarantees, never worse than the baseline and no
componentwise regression, and finds real improvements; (3) the choice of t-norm
is empirically irrelevant on the standard benchmark and provably relevant on an
adversarial suite; (4) certification caught two real defects, one in a widely
used upstream library and one in our own cross-platform reproducibility. All
artifacts, seeds, and scripts are released.

\subsection{Benchmark and gates}\label{sec:bench}

The benchmark comprises 100 circuits: 62 real circuits, namely 26 native and 25
transpiled variants from the QASMBench small category \cite{qasmbench} (commit
\texttt{357b9423}; transpilation via a Qiskit preset pass) and 11 unitary
circuits from the Feynman repository \cite{amy2018,feynman} (commit
\texttt{d2c382a2}), plus 38 curated generated circuits (seed 42): QFT on
3/5/6/7/8 qubits, Cuccaro and VBE adders, QAOA MaxCut at $p = 1..3$, UCC-style
VQE ans\"atze, Grover instances for $n = 3, 4, 5$, a functionally verified
$2\times2$ multiplier and a 3-bit floor-square-root, a 7-T Clifford+T Toffoli
decomposition and a Toffoli chain, and random symplectic Clifford+T circuits of
varying depths (20--80 gates, 3--8 qubits).

Selection criteria: 3--10 qubits, 3--200 gates, parseability by PyZX 0.10.4, a
stable QASM round-trip, and purely unitary circuits, which the certifier
requires. Because QASMBench circuits ship with terminal measurements, and some
with resets or classically controlled gates, a documented preprocessing step
strips them to their unitary kernel (\texttt{stripped\_nonunitary: true} in the
manifest; 51 circuits affected). Fixed global seed 42; every circuit is
optimized by the PyZX baseline and by the agent under each of the three t-norms,
for 400 optimizations in total. All artifacts (100 QASM files, manifest with
provenance, results, 400 certificates) are released with the repository.

Three gates guard the claims, mirroring the gated methodology of \cite{lggt}.

\paragraph{Gate G0 (metric ground truth): PASS, 10/10.}
Published per-circuit QASMBench/Feynman metrics were not locally available, so
an independent second code path, a hand-written QASM parser computing
\texttt{t\_count}, \texttt{two\_qubit\_count} and \texttt{depth} directly from
text without PyZX objects, was compared against AlchemQ's metric module on 10
real reference circuits. All 10 circuits $\times$ 3 metrics agreed to the unit.

\paragraph{Gate G1 (certifier completeness and soundness): PASS, 100/100 and
100/100.}
On 200 pairs built from 2--5-qubit base circuits (seed 424242): 100 equivalent
pairs, each with a known identity ($H^2$, $X^2$, $\mathrm{CNOT}^2$,
$TT^{\dagger}$, $SS^{\dagger}$, $P(\pi/8)P(-\pi/8)$) inserted at a random
position, were certified 100/100, and 100 mutated pairs (one extra $X$/$Z$/$T$
gate or a $T \to S$ swap) were rejected 100/100. Zero false positives, zero
false negatives on this class.

\paragraph{Gate G2 (agent non-inferiority on a dev subset): PASS, 10/10, with an
honesty caveat.}
On a random dev subset of 10 circuits (seed 7) the agent's G\"odel aggregate is
$\geq$ the baseline's on 10/10. We disclosed in v0.2 that on all 10 dev circuits
both aggregates are 0.0, so the gate passes by construction; recomputed under
the current zero-baseline score semantics, one dev circuit
(\texttt{error\_correctiond3\_n5}) shows a genuine nonzero advantage (0.755
versus 0.694) and the remaining nine remain $0 = 0$. The agent's actual
advantages appear on 9/100 circuits of the full set (Table~\ref{tab:wins}).

Total wall time for the full run was 539.5\,s on a standard workstation,
recorded in \pth{benchmark/full\_results.json}.

\subsection{Main results (archived run, recomputed)}

The archived benchmark run was produced by the v0.2 optimizer and executes 100
circuits $\times$ 4 methods $=$ 400 optimizations: \textbf{0 timeouts, 0 errors,
0 uncertified results}. We report its circuit-level outcomes and recompute all
aggregates from the archived metrics under the current score semantics
\eqref{eq:score}. The zero-baseline revision makes Clifford-heavy circuits
visible to the aggregate, which lifts the means but not the ordering
conclusions; Table~\ref{tab:agg} shows both. Aggregates of different norms are
not on a common scale: each column compares agent against baseline within one
norm.

Under the current semantics the agent is strictly better than the baseline on
9/100 circuits under $T_G$ and $T_P$, and on 2/100 under $T_L$ (the archived
values were 8/100 and 1/100; the new win is the Clifford circuit
\texttt{error\_correctiond3\_n5}, invisible before the revision). The agent is
never worse than the baseline under any norm (0/300 agent runs), confirming
Proposition~1(ii) empirically. The three norms return identical final circuits
on all 100 instances, with certificate ids coinciding for the three agents on
every circuit, so on this benchmark the t-norm choice has no observable effect;
Section~\ref{sec:ablation} shows where that breaks.

Over the 82 circuits with nonzero baseline T-count, the baseline alone reduces
the T-count by 21.4\% on average, a strong baseline. The certification methods
used for the 400 results: \texttt{zx-full-reduce} in 384 cases,
\texttt{numeric-tensor} in 16 (baseline-only split: 96 and 4).

\begin{table}[t]
\centering
\small
\caption{Mean aggregates over the 100 benchmark circuits, recomputed from
\pth{benchmark/full\_results.json} under the v0.4 score semantics
\eqref{eq:score}; archived (v0.2) semantics in parentheses. The agent is never
below the baseline on any circuit under either semantics.}
\label{tab:agg}
\begin{tabular}{@{}lll@{}}
\toprule
Method & Mean aggregate & Strict wins vs. baseline \\
\midrule
baseline (PyZX \texttt{full\_reduce})
  & $0.0490\,T_G$ / $0.0397\,T_P$ / $0.0346\,T_L$ & --- \\
  & \emph{archived:} $0.0137\,T_G$ / $0.0057\,T_P$ / $0.0013\,T_L$ & \\
agent-G\"odel        & 0.0584 (arch. 0.0225) & 9/100 (arch. 8/100) \\
agent-product        & 0.0450 (arch. 0.0103) & 9/100 (arch. 8/100) \\
agent-{\L}ukasiewicz & 0.0393 (arch. 0.0050) & 2/100 (arch. 1/100) \\
\bottomrule
\end{tabular}
\end{table}

\begin{table}[t]
\centering
\small
\caption{Top strict wins (agent-G\"odel vs. baseline), recomputed under the
current score semantics. Metrics are original $\rightarrow$ agent (baseline in
parentheses where it differs).}
\label{tab:wins}
\begin{tabular}{@{}lrrll@{}}
\toprule
Circuit & $A_{T_G}$ agent & $A_{T_G}$ base
        & T: orig $\rightarrow$ agent (base)
        & depth: orig $\rightarrow$ agent (base) \\
\midrule
\texttt{error\_correctiond3\_n5} & 0.755 & 0.694 & $0 \to 0$    & $77 \to 12$ (15) \\
\texttt{simon\_n6} (transpiled)  & 0.714 & 0.429 & $14 \to 0$ (0) & $33 \to 7$ (10) \\
\texttt{rand\_ct\_n3\_g40}       & 0.440 & 0.400 & $17 \to 5$ (5) & $25 \to 14$ (15) \\
\texttt{pea\_n5} (transpiled)    & 0.333 & 0.286 & $48 \to 15$ (15) & $87 \to 40$ (40) \\
\texttt{rand\_ct\_n6\_g50}       & 0.273 & 0.091 & $24 \to 10$ (10) & $20 \to 14$ (17) \\
\texttt{rand\_ct\_n4\_g60}       & 0.167 & 0.067 & $20 \to 8$ (8) & $30 \to 25$ (28) \\
\texttt{rand\_ct\_n5\_g70}       & 0.150 & 0.100 & $32 \to 8$ (8) & $30 \to 24$ (24) \\
\texttt{simon\_n6}               & 0.125 & 0.000 & $14 \to 0$ (0) & $8 \to 7$ (10) \\
\texttt{rand\_ct\_n7\_g60}       & 0.045 & 0.000 & $26 \to 6$ (6) & $27 \to 22$ (26) \\
\bottomrule
\end{tabular}
\end{table}

\paragraph{Non-regression then and now.}
The archived run predates the v0.4 mechanism of Proposition~1(iii): 42/100 agent
outputs of that run exhibit a componentwise regression against the original,
every one at aggregate 0 (Remark~3), that is, ties with the origin resolved in
favour of a regressing candidate by the old tie-break. Under the current
optimizer those ties resolve to non-regressing results with identical
aggregates; we validate this directly on the standard suite below.

\subsection{Standard suite after the v0.4 non-regression fix}

The 16-circuit standard suite (\pth{results/results.json}) was re-generated with
the v0.4 optimizer (benchmark run twice, identical outputs). Mean aggregates are
exactly unchanged versus the archived run: G\"odel, agent 0.046845 versus
baseline 0.038095; product, 0.014545 versus 0.010710; {\L}ukasiewicz, 0 versus
0. But the returned circuit changed on 9/16 instances
(Table~\ref{tab:nonregress}): all nine previously returned a regressing
tie-winner and now return the original circuit. In total 13/16 agent outputs
equal the original circuit, 3/16 are genuine improvements (the same three as
before the fix), and 0/16 regress any objective, so Proposition~1(iii) holds
empirically. All 48 agent certificates use \texttt{zx-full-reduce}, and the
three norms still return identical circuits on 16/16 instances.

\begin{table}[t]
\centering
\small
\caption{The nine standard-suite circuits whose returned circuit changed with
the v0.4 non-regression rank (metrics $=$ T / two-qubit / depth). ``pre-v0.4''
is the archived winner; ``v0.4+'' is the current winner, equal to the original
in all nine cases. Aggregates are unchanged (all 0).}
\label{tab:nonregress}
\begin{tabular}{@{}lccc@{}}
\toprule
Circuit & original & pre-v0.4 winner & v0.4+ winner \\
\midrule
\texttt{qft5}             & 20/10/9  & 18/33/52 & 20/10/9 \\
\texttt{qft6}             & 30/15/11 & 25/48/70 & 30/15/11 \\
\texttt{toffoli\_chain}   & 14/12/20 & 12/21/36 & 14/12/20 \\
\texttt{adder2}           & 28/33/50 & 16/46/58 & 28/33/50 \\
\texttt{rand\_ct\_n3\_g20} & 9/7/13   & 3/5/14   & 9/7/13 \\
\texttt{rand\_ct\_n4\_g20} & 9/3/9    & 3/3/10   & 9/3/9 \\
\texttt{rand\_ct\_n5\_g30} & 7/10/12  & 5/12/16  & 7/10/12 \\
\texttt{rand\_ct\_n5\_g50} & 30/9/23  & 6/12/15  & 30/9/23 \\
\texttt{rand\_ct\_n6\_g40} & 23/11/13 & 5/15/19  & 23/11/13 \\
\bottomrule
\end{tabular}
\end{table}

\subsection{Artifact validation and cross-platform reproducibility}
\label{sec:artifacts}

Three certificate corpora are reported here: the 400 legacy protocol-0.1.1
certificates from the benchmark run (\pth{benchmark/certs/}), the 400 historical
protocol-0.2.0 certificates migrated at v0.4 with \texttt{phase-note-v1} notes
(\pth{evidence/v04/certs/}), and 400 fresh protocol-0.2.0 certificates with
\texttt{phase-note-v2} notes (\pth{evidence/v05/certs/}).
Table~\ref{tab:validation} summarizes the validation gates, each executed on two
platforms (Linux, Python 3.12, NumPy 2.2.5, PyZX 0.10.4; macOS, Python 3.14.4,
NumPy 2.5.1, PyZX 0.10.4). The full 2998-test suite passes in approximately
80\,s on both platforms; tamper detection is 100\% across the 260-case V1
mutation matrix and the 27 V2 mutation cases.

\begin{table}[t]
\centering
\small
\caption{Artifact validation gates (all pass on both platforms). ``v1 boolean''
is \texttt{verify\_certificate}; ``v2 detailed'' is
\texttt{verify\_certificate\_v2\_detailed}. The legacy corpus is the archived
2026-06 benchmark run; the v0.4 corpus is migrated from it; the determinism
audit compares two independent generation runs.}
\label{tab:validation}
\begin{tabular}{@{}llll@{}}
\toprule
Artifact set & Protocol & Verifier & Result \\
\midrule
legacy benchmark corpus & 0.1.1 & v1 boolean & 400/400 \\
historical v0.4 corpus  & 0.2.0, \texttt{phase-note-v1} & v2 detailed & 400/400 certified \\
fresh v0.5 corpus       & 0.2.0, \texttt{phase-note-v2} & v2 detailed & 400/400 certified \\
determinism audit       & 0.2.0 & byte diff, excl. timestamp & 400/400 identical \\
\bottomrule
\end{tabular}
\end{table}

Seven legacy certificates deserve a remark: they do not reproduce under the
declared legacy engine, for known reasons. Six rendered a raw phase float that
wrapped around $2\pi$; one is the \texttt{variational\_n4} case study below. The
v1 verifier accepts them through a scoped cross-engine fallback, in which the
declared engine's verdict is corroborated by the other engine (ZX-declared and
tensor-corroborated, or conversely), and the V2 migration re-issues them with
canonical notes. All seven validate under both protocols.

\subsection{Case study 1: a false negative that was a PyZX bug}
\label{sec:case1}

The most instructive failure in the project's history is the certificate for
\pth{qasmbench\_\_variational\_n4\_\_transpiled} (baseline). Under the pre-v0.5
certifier its numeric fallback delegated the decision to PyZX's
\texttt{compare\_tensors}, which normalizes both tensors by the first element of
modulus $> 10^{-14}$ and then compares. For this 4-qubit pair the pivot element
has amplitude $4.7\times10^{-9}$, right at the cancellation floor and against
typical amplitudes ${\sim}\,1/16$, so the normalization amplifies
floating-point noise by ${\sim}\,2\times10^{8}$ and the normalized tensors
differ by up to $1.66\times10^{6}$. PyZX answered ``not equal''; the certificate
was a false negative, and a correct optimization looked unverifiable.

The optimal-overlap residual of the same pair is $1.67\times10^{-10}$ (fresh
v0.5 measurement: $7.4\times10^{-11}$, matching an independent dense-matrix
comparison at $\texttt{max\_error} = 7.4\times10^{-11}$), eleven orders of
magnitude below any reasonable threshold.

The fix, in module \texttt{equivalence.py} shared by both protocols, is the
overlap criterion of Section~\ref{sec:engine}: the phase is estimated from all
entries via the Hilbert--Schmidt inner product, never from a single pivot. The
pair now certifies under both protocols (\texttt{numeric-tensor}, residual
$7.4\times10^{-11}$), and the legacy certificate validates through the
cross-engine fallback above.

The certification layer rejected a correct optimization. Investigation traced
the cause to PyZX's pivot normalization rather than to the rewrite itself. This
is the behaviour a certification layer is built to produce.

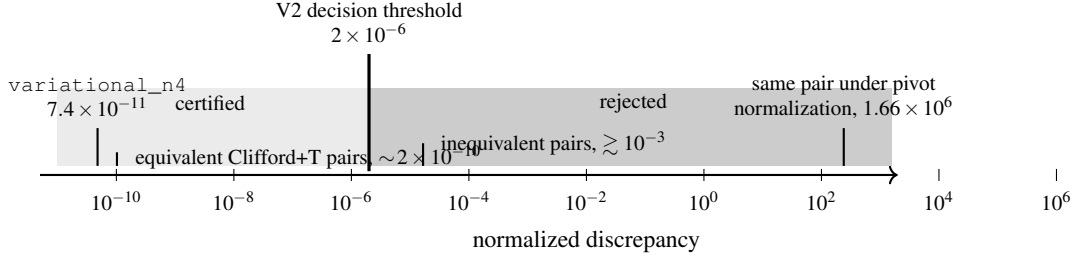
\begin{figure}[t]
\centering
\begin{tikzpicture}[x=0.777cm, y=1cm]
% axis
\draw[->,thick] (-0.3,0) -- (14.3,0);
\foreach \e/\x in {-10/1,-8/3,-6/5,-4/7,-2/9,0/11,2/13,4/15,6/17}
  \draw (\x,0.08) -- (\x,-0.08) node[below,font=\scriptsize] {$10^{\e}$};
\node[below=4mm] at (9,-0.2) {\footnotesize normalized discrepancy};
% regions
\fill[black!8] (0,0.12) rectangle (5.3,1.15);
\fill[black!20] (5.3,0.12) rectangle (14.2,1.15);
\node[font=\scriptsize] at (2.6,0.95) {certified};
\node[font=\scriptsize] at (9.8,0.95) {rejected};
% threshold
\draw[very thick] (5.3,0.05) -- (5.3,1.6);
\node[font=\scriptsize,align=center,above] at (5.3,1.6)
  {V2 decision threshold\\ $2\times10^{-6}$};
% markers
\draw[thick] (0.68,0.12) -- (0.68,0.62);
\node[font=\scriptsize,align=center,above,rotate=0] at (0.68,0.62)
  {\texttt{variational\_n4}\\ $7.4\times10^{-11}$};
\draw[thick] (1.01,0.12) -- (1.01,0.30);
\node[font=\scriptsize,right] at (1.15,0.24)
  {equivalent Clifford+T pairs, ${\sim}\,2\times10^{-10}$};
\draw[thick] (6.22,0.12) -- (6.22,0.42);
\node[font=\scriptsize,right] at (6.35,0.42)
  {inequivalent pairs, $\gtrsim 10^{-3}$};
\draw[thick] (13.38,0.12) -- (13.38,0.62);
\node[font=\scriptsize,align=center,above] at (13.38,0.62)
  {same pair under pivot\\ normalization, $1.66\times10^{6}$};
\end{tikzpicture}
\caption{The decision scale of the numeric path, drawn from the values reported
in Sections~\ref{sec:engine} and \ref{sec:case1}. Equivalent and inequivalent
Clifford+T pairs sit about four orders of magnitude on either side of the
$2\times10^{-6}$ threshold. The \texttt{variational\_n4} pair, which PyZX's
pivot normalization placed at $1.66\times10^{6}$ and therefore rejected, has an
optimal-overlap residual of $7.4\times10^{-11}$: sixteen orders of magnitude
apart, on the same pair of circuits.}
\label{fig:scale}
\end{figure}

\subsection{Case study 2: eight certificates that failed only on macOS}
\label{sec:case2}

The v0.4 migrated corpus validated 400/400 on Linux (Python 3.12) but only
392/400 on macOS (Python 3.14). The eight rejections, namely
\texttt{bell\_n4} (baseline and three agents), \texttt{linearsolver\_n3} (three
agents), and the \texttt{variational\_n4} baseline, were all
\texttt{numeric-tensor} certificates whose \texttt{phase\_note} rendered raw
floats, for example the phase angle, with platform-dependent low-order digits.
The numbers come out of BLAS/LAPACK eigen- and dot-kernels whose summation order
differs across builds, and the byte-exact note comparison then failed even
though every mathematical claim reproduced perfectly.

The fix is the versioned phase-note machinery of Section~\ref{sec:v2}.
\texttt{phase-note-v1} notes are validated semantically (grammar, canonical
phase equality modulo $2\pi$, residual consistency within
\texttt{RESIDUAL\_EQ\_TOL}, rendering within half a quantum of the canonical
form), which absorbs exactly this platform noise; \texttt{phase-note-v2} notes
are rendered from the canonical phase representative and are byte-reproducible
across platforms by construction. Unknown schema versions yield
\texttt{inconclusive}. After the fix the historical corpus validates 400/400 on
both platforms, the fresh corpus is byte-deterministic across two independent
generations (400/400, timestamp excluded), and the compatibility suite includes
62 dedicated platform/tamper tests.

\subsection{T-norm ablation: does the norm choice matter?}\label{sec:ablation}

On the standard set, no. On the 100-circuit benchmark the three agent variants
returned the identical final circuit (same certificate hash) in 100/100 cases;
on the 16-circuit standard suite the same holds 16/16 after the v0.4
regeneration. This confirms, at community-benchmark scale, that multi-objective
aggregation, while implemented and axiomatically tested, does not differentiate
the search outcome on simple circuits.

\paragraph{Why: the dead-zone observability argument.}
Because scores are clamped in \eqref{eq:score}, a candidate that fails to
improve an objective with a positive baseline has a zero component there, hence
$A_T = 0$ under all three norms. Norm-divergence therefore requires a Pareto
front containing at least two points that improve all three objectives with
different profiles, which is empirically rare on shallow circuits, whose fronts
are often singletons or entirely zero-scored. This motivated an adversarial
suite.

\paragraph{The 38-circuit ablation suite.}
We assembled a 38-circuit suite (the 16-circuit v0.1 suite, 14 additional
structured circuits from QFT, Toffoli-cascade, ansatz and QAOA families, and 8
dense adversarial random Clifford+T circuits of 70--120 gates engineered to
produce multi-point all-improving fronts; seed 42) and ran the instrumented
agent once per norm. For each circuit we pool all certified intermediate
candidates, compute the Pareto front over (T, 2Q, depth), and define the winner
of norm $T$ as the front point maximizing its aggregate with a deterministic
tie-break (sum of scores, then lower T, 2Q, depth). A circuit is
\emph{norm-divergent} iff the three norms do not all pick the same winner. The
suite is untouched by the v0.4 RNG changes, since only the ansatz/QAOA/Toffoli
streams moved and not the dense family, and the ablation artifacts predate v0.4
but are unaffected by its score revision, since every instance has positive
baselines on all three objectives.

\paragraph{Results: norm-divergence is real but structurally rare.}
Table~\ref{tab:ablation} reports the pairwise outcome. 4/38 circuits are
norm-divergent, 3/38 are product-versus-G\"odel divergent, which is the
substantively interesting case since $T_L$ is often zero-informative, and all
divergent cases lie in the dense adversarial family
(\texttt{rand\_ct\_dense}, 8 circuits: 4 norm-divergent, 3 P-vs-G). Every
discordant pair is Pareto-incomparable: neither norm's final circuit dominates
the other's on (T, 2Q, depth), so the exact McNemar test has no dominance signal
and returns $p = 1.0$ for all three pairs. The norms disagree, but no norm wins
systematically.

\begin{table}[t]
\centering
\small
\caption{Pairwise ablation statistics on the 38-circuit suite (final circuits
per norm compared by QASM hash; direction $=$ which norm's circuit
Pareto-dominates; incomparable pairs are excluded from the binomial test). Exact
two-sided McNemar $p$-values \cite{mcnemar}.}
\label{tab:ablation}
\begin{tabular}{@{}lrrrrr@{}}
\toprule
Pair & Concordant & Discordant & Dominates A & Dominates B & Incomparable (exact $p$) \\
\midrule
$T_L$ vs $T_P$ & 37 & 1 & 0 & 0 & 1 ($p = 1.0$) \\
$T_L$ vs $T_G$ & 34 & 4 & 0 & 0 & 4 ($p = 1.0$) \\
$T_P$ vs $T_G$ & 35 & 3 & 0 & 0 & 3 ($p = 1.0$) \\
\bottomrule
\end{tabular}
\end{table}

\begin{table}[t]
\centering
\small
\caption{The four norm-divergent circuits (all in the dense adversarial family):
Pareto-front winners per norm as (T, 2Q, depth). The last row diverges only
through the $T_L$ tie-break.}
\label{tab:divergent}
\begin{tabular}{@{}lccc@{}}
\toprule
Circuit & Winner $T_L$ & Winner $T_P$ & Winner $T_G$ \\
\midrule
\pth{rand\_ct\_dense\_n3\_g80\_s42\_divp}  & (10, 11, 26) & (10, 11, 26) & (10, 12, 24) \\
\pth{rand\_ct\_dense\_n4\_g110\_s42\_divp} & (10, 25, 36) & (10, 25, 36) & (10, 24, 39) \\
\pth{rand\_ct\_dense\_n5\_g100\_s46\_divp} & (14, 27, 39) & (14, 27, 39) & (14, 26, 43) \\
\pth{rand\_ct\_dense\_n4\_g100\_s44\_divl} & (13, 22, 31) & (13, 21, 33) & (13, 21, 33) \\
\bottomrule
\end{tabular}
\end{table}

Table~\ref{tab:divergent} shows the mechanism concretely. On
\pth{rand\_ct\_dense\_n3\_g80\_s42\_divp} the front contains two all-improving points:
(10, 11, 26), with better two-qubit count and worse depth, and (10, 12, 24), the
reverse. $T_G$ (bottleneck, maximin) prefers the balanced profile (10, 12, 24)
with aggregate 0.478; $T_P$ prefers (10, 11, 26) with aggregate 0.177, because
its score product is higher there; the two winners are Pareto-incomparable. This
is precisely the predicted regime: norm choice matters exactly when the front
offers a genuine trade-off between objectives that all improved.

\begin{figure}[t]
\centering
\begin{tikzpicture}[x=3cm, y=1cm]
\draw[->,thick] (-0.15,0) -- (2.35,0)
  node[right,font=\scriptsize] {two-qubit count};
\draw[->,thick] (0,-0.3) -- (0,4.4)
  node[above,font=\scriptsize] {depth};
\foreach \v/\x in {11/0.5, 12/1.5}
  \draw (\x,0.06) -- (\x,-0.06) node[below,font=\scriptsize] {\v};
\foreach \v/\y in {24/1, 25/2, 26/3}
  \draw (0.02,\y) -- (-0.02,\y) node[left,font=\scriptsize] {\v};
% front
\draw[dashed,thick,black!50] (0.5,3) -- (1.5,1);
\fill (0.5,3) circle (2.6pt);
\fill (1.5,1) circle (2.6pt);
\node[font=\scriptsize,align=center,above right] at (0.55,3.05)
  {$(10,11,26)$\\ chosen by $T_L$, $T_P$\\ $A_{T_P} = 0.177$};
\node[font=\scriptsize,align=center,below left] at (1.45,0.92)
  {$(10,12,24)$\\ chosen by $T_G$\\ $A_{T_G} = 0.478$};
\end{tikzpicture}
\caption{The Pareto front of \pth{rand\_ct\_dense\_n3\_g80\_s42\_divp}, the
mechanism behind Table~\ref{tab:divergent}. Both points improve all three
objectives and both have T-count 10, so neither dominates the other. $T_G$
(maximin) takes the balanced profile; $T_P$ takes the one with the higher score
product. Norm choice matters exactly here, where the front offers a genuine
trade-off between objectives that all improved.}
\label{fig:pareto}
\end{figure}
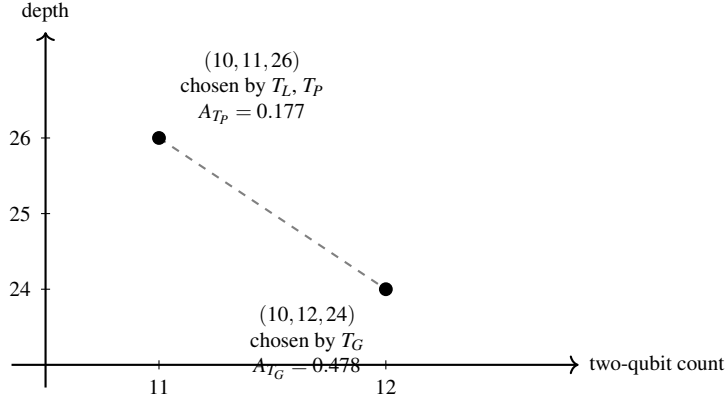

\paragraph{{\L}ukasiewicz is zero-informative, and we can say why.}
On 36/38 circuits of the ablation suite the $T_L$ aggregate is 0.0 on every
front point. This is not a bug but the dead zone of the nilpotent norm: with
$n = 3$ objectives, $A_{T_L} = 0$ unless $s_1 + s_2 + s_3 > 2$, that is, unless
a candidate nearly eliminates two of three objectives at once. Such near-total
joint improvements essentially do not occur for ZX-based simplification of
Clifford+T circuits, where gains in T-count are typically paid for in depth or
two-qubit count. The norm remains correct, being axiom-tested, and registers
signal exactly where the semantics says it should: the two dense control
circuits where candidates do clear the threshold
(\pth{rand\_ct\_dense\_n3\_g100\_s42\_ctl} with $A_{T_L} = 0.101$ and
\pth{rand\_ct\_dense\_n4\_g80\_s42\_ctl} with $A_{T_L} = 0.081$). For learning,
however, a flat dead zone means zero gradient almost everywhere, a concrete
argument against $T_L$ as reward shaping for circuit optimization, and for
$T_P$/$T_G$ as the practical choices.

\paragraph{Runtime.}
Per-circuit wall time (optimization plus certification) averages 0.11\,s for the
baseline (max 1.14\,s) and approximately 2.7\,s per agent run (max 31\,s),
summed over the 100 circuits to per-method totals of 11.5\,s (baseline) and
274.1/272.3/271.2\,s (agents under the three norms); the full benchmark run took
539.5\,s wall time including mutation testing. Certification is dominated by the
ZX path and is cheap relative to search; its cost is paid once per artifact
rather than once per consumer.

\subsection{First hardware experiment: IBM Heron r2}\label{sec:qpu}

\paragraph{Setup.}
On 2026-07-24 we ran the certified \pth{qasmbench\_\_simon\_n6\_\_transpiled} pair
(the certificate of Listing~\ref{lst:cert}, re-verified immediately before
submission) on \texttt{ibm\_kingston}, a 156-qubit Heron~r2 processor: both
circuits, 1024 shots each, within a single job
(\texttt{d9hjh7t0k0jc738i9phg}) to minimize temporal and calibration drift
between the two executions. Both circuits were independently mapped to the
backend ISA using Qiskit's preset pass manager
(\pth{generate\_preset\_pass\_manager}) with
\texttt{optimization\_level=1}, the same fixed initial layout on physical qubits
0--5, trivial initial-layout selection, and SABRE routing. The backend target
and calibration data available at transpilation time were identical for both
arms. The comparison therefore measures the effect of AlchemQ preprocessing
followed by the same downstream Qiskit hardware-mapping configuration; it does
not isolate AlchemQ from all compiler-side transformations. Qiskit 2.5.1 and
\texttt{qiskit-ibm-runtime} 0.48.0 were used for submission and result
retrieval.

\paragraph{ISA cost.}
Compiled to native instructions, the certified circuit is dramatically cheaper:
depth $110 \to 24$ ($-78.2\%$), size $159 \to 54$, CZ gates $29 \to 10$
($-65.5\%$), SX $60 \to 18$, RZ $63 \to 20$. On Heron-class hardware, where
two-qubit gates dominate the error budget, this is the practically relevant
saving, larger than the logical-level metrics alone suggest.

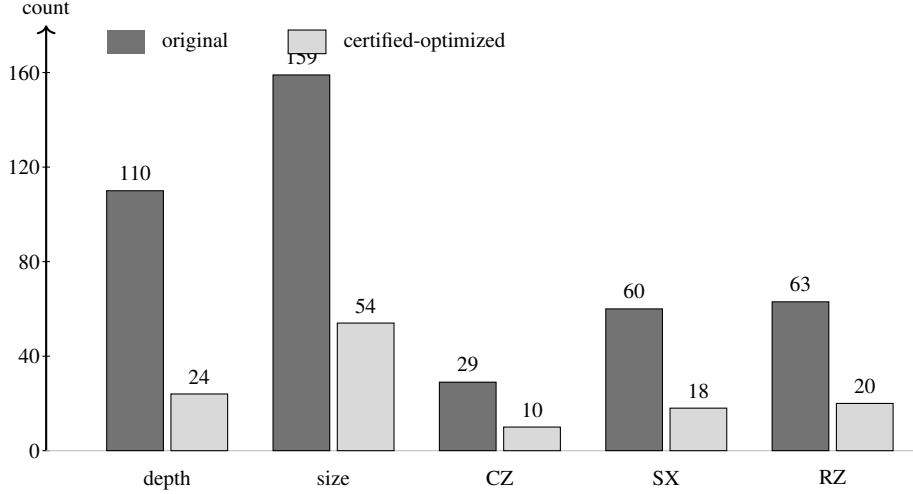
\begin{figure}[t]
\centering
\begin{tikzpicture}[y=0.03125cm]
\draw[->,thick] (-0.4,0) -- (-0.4,180) node[above,font=\scriptsize] {count};
\foreach \v in {0,40,80,120,160}
  \draw (-0.45,\v) -- (-0.35,\v) node[left,font=\scriptsize] {\v};
\draw[black!25] (-0.4,0) -- (11.2,0);
% groups: name / original / optimized / x-offset
\foreach \nm/\o/\p/\x in {%
  depth/110/24/0.4, size/159/54/2.6, CZ/29/10/4.8, SX/60/18/7.0, RZ/63/20/9.2}{
  \fill[black!55] (\x,0)      rectangle ++(0.75,\o);
  \fill[black!15] (\x+0.85,0) rectangle ++(0.75,\p);
  \draw (\x,0)      rectangle ++(0.75,\o);
  \draw (\x+0.85,0) rectangle ++(0.75,\p);
  \node[font=\scriptsize,above] at (\x+0.375,\o) {\o};
  \node[font=\scriptsize,above] at (\x+1.225,\p) {\p};
  \node[font=\scriptsize,below] at (\x+0.8,-4) {\nm};
}
\fill[black!55] (0.4,168) rectangle ++(0.5,10);
\node[font=\scriptsize,right] at (1.0,173) {original};
\fill[black!15] (2.8,168) rectangle ++(0.5,10);
\draw (2.8,168) rectangle ++(0.5,10);
\node[font=\scriptsize,right] at (3.4,173) {certified-optimized};
\end{tikzpicture}
\caption{ISA-level cost of the two arms of the hardware experiment
(Section~\ref{sec:qpu}), compiled to the native gate set of
\texttt{ibm\_kingston} under an identical downstream Qiskit configuration.
Depth falls by 78.2\% and CZ count by 65.5\%. On Heron-class hardware, where
two-qubit gates dominate the error budget, the CZ column is the practically
relevant one.}
\label{fig:isa}
\end{figure}

\paragraph{Output quality.}
Both output distributions were compared against the exact ideal distribution
(statevector simulation of the original circuit): total-variation distance
0.1143 (original) versus 0.1035 (optimized); Hellinger fidelity 0.9069 versus
0.9142; and the fraction of outcomes satisfying the recovered hidden Simon
constraints ($y \cdot s = 0 \pmod 2$; support of sixteen states) 91.31\% versus
91.80\%. All three metrics favour the optimized circuit. The bootstrap 95\%
confidence intervals, however, include zero:
$\Delta_{\text{quality}} = +0.0107$ with CI $[-0.0215, +0.0459]$ for the TVD
difference, and $+0.0073$ with CI $[-0.0184, +0.0296]$ for fidelity. At 1024
shots the shot-noise floor of the TVD estimator alone is 0.0492 in the mean and
0.0645 at the 95th percentile (Monte-Carlo on the 16-state uniform ideal), so
the quality trend is consistent in direction but not statistically resolved.

\begin{figure}[t]
\centering
\begin{tikzpicture}[x=1cm, y=1cm]
\fill[black!8] (1.486,0.6) rectangle (8.514,3.9);
\draw[black!35,dashed] (9.607,0.6) -- (9.607,3.9);
\draw[black!35,dashed] (0.393,0.6) -- (0.393,3.9);
\node[font=\scriptsize,black!55,align=center,above] at (5.0,3.9) {shot-noise floor\\ mean $0.0492$, p95 $0.0645$};
\draw[very thick] (5.0,0.5) -- (5.0,4.0);
\draw[thick] (3.464,3.2) -- (8.279,3.2);
\draw[thick] (3.464,3.08) -- (3.464,3.3200000000000003);
\draw[thick] (8.279,3.08) -- (8.279,3.3200000000000003);
\fill (5.764,3.2) circle (2.4pt);
\node[font=\scriptsize,left] at (-0.25,3.2) {total-variation distance};
\draw[thick] (3.686,2.2) -- (7.114,2.2);
\draw[thick] (3.686,2.08) -- (3.686,2.3200000000000003);
\draw[thick] (7.114,2.08) -- (7.114,2.3200000000000003);
\fill (5.522,2.2) circle (2.4pt);
\node[font=\scriptsize,left] at (-0.25,2.2) {Hellinger fidelity};
\fill (5.35,1.2) circle (2.4pt);
\node[font=\scriptsize,left] at (-0.25,1.2) {Simon constraint rate};
\draw[->,thick] (-0.2,0.5) -- (10.4,0.5);
\draw (0.714,0.5) -- (0.714,0.42) node[below,font=\scriptsize] {$-0.06$};
\draw (2.143,0.5) -- (2.143,0.42) node[below,font=\scriptsize] {$-0.04$};
\draw (3.572,0.5) -- (3.572,0.42) node[below,font=\scriptsize] {$-0.02$};
\draw (5.0,0.5) -- (5.0,0.42) node[below,font=\scriptsize] {$0$};
\draw (6.429,0.5) -- (6.429,0.42) node[below,font=\scriptsize] {$+0.02$};
\draw (7.857,0.5) -- (7.857,0.42) node[below,font=\scriptsize] {$+0.04$};
\draw (9.286,0.5) -- (9.286,0.42) node[below,font=\scriptsize] {$+0.06$};
\node[font=\scriptsize,below=4mm] at (5,0.42)
  {difference in favour of the certified-optimized circuit};
\end{tikzpicture}
\caption{Output quality on \texttt{ibm\_kingston}, 1024 shots per arm, as
reported in Section~\ref{sec:qpu}. All three point estimates favour the
optimized circuit, and every bootstrap 95\% confidence interval includes zero.
The effect is roughly a fifth of the shot-noise floor of the TVD estimator at
this shot count, so the trend is consistent in direction but not statistically
resolved. The Simon constraint rate is shown as a point estimate; no bootstrap
interval is reported for it.}
\label{fig:quality}
\end{figure}
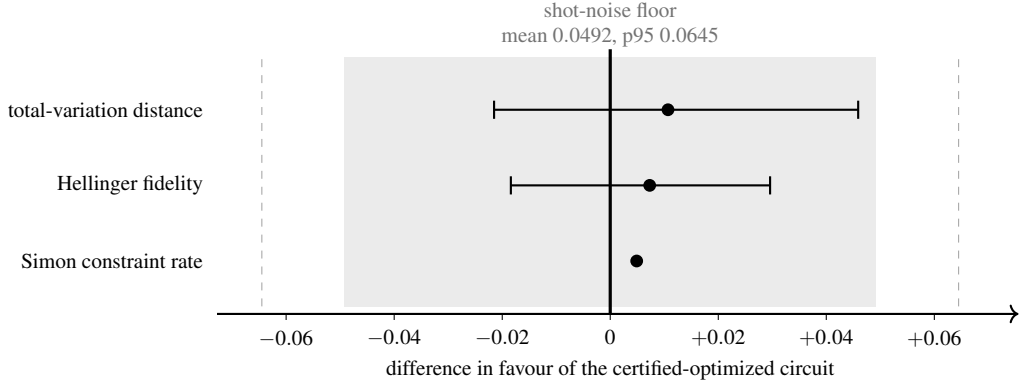

\paragraph{Reading.}
The experiment supports one claim today: the certificate's promise survives
contact with real hardware. A 78\% shallower circuit with 65\% fewer two-qubit
gates produces measured output with no detectable degradation at this shot
count, with a consistent but unproven trend toward better. Turning the trend
into a significance result requires an interleaved, multi-session protocol (at
least three calibration sessions, approximately $6\times10^{3}$ shots per arm per
session, bias-corrected TVD, calibration-snapshot provenance), which is ongoing work. We
report the insignificant pilot rather than wait for it.

\section{Discussion}\label{sec:discussion}

\paragraph{Why a proof-carrying optimizer.}
The design philosophy is adversarial toward the optimizer itself. Anything that
produces a circuit is untrusted, whether it is a beam search today, an RL policy
tomorrow, or an LLM-proposed rewrite the year after. The certificate is the only
interface and the verifier is the only trusted component. This inverts the
verified-compiler trade-off. Instead of proving the optimizer correct, which is
hard and tied to one optimizer, we make its outputs independently checkable,
which is easy per artifact and optimizer-agnostic.

The v0.5 history supports the stance twice. Certification surfaced a
normalization bug in PyZX's \texttt{compare\_tensors} that a test suite keyed on
PyZX alone could not have caught (Section~\ref{sec:case1}). Byte-level
certificate reproducibility surfaced a cross-platform BLAS dependence that
purely mathematical validation would have passed (Section~\ref{sec:case2}). In
both cases the audit surface was the artifact rather than the code.

\paragraph{What the agent is for.}
The toy agent exercises the contract under realistic load: 400 optimizations,
three aggregation norms, a strong baseline, mutation testing, and an adversarial
ablation suite. Two of its findings are useful beyond this system. The t-norm
irrelevance on standard benchmarks, sharpened by the non-regression rank, tells
future RL work that the aggregation objective is not the knob that matters on
these distributions. The regression-tie phenomenon tells future work that
\emph{which} certified circuit is returned at aggregate zero is a policy
question worth learning. The agent is simple by construction and all performance
claims are scoped accordingly.

\paragraph{T-norms as an acceptance family, not an objective.}
The three t-norms act as acceptance criteria sharing one structural property:
the annihilator at zero turns any componentwise regression into aggregate zero,
which is what makes the safety proofs go through. They are not a defuzzified
objective to be maximized blindly. The G\"odel (minimum) norm in particular is
known in the fuzzy literature to be Pareto-suboptimal as an aggregator, because
it ignores all non-minimal components. The empirical finding that all three
norms return identical circuits on the standard benchmark
(Section~\ref{sec:ablation}) shows that on these distributions the choice of
norm is not where the action is.

A principled selection layer is left to future work: a certified Pareto front
with per-point certificates, calibration-aware risk terms carrying explicit
uncertainty, and selection profiles (FTQC, NISQ, fidelity) built on top of the
unchanged equivalence contract, rather than a single scalar winner.

\paragraph{The certificate as a reproducibility object.}
Protocol 0.2.0 turns the certificate into a release artifact. Canonical hashes
make it tamper-evident. Versioned schemas make it reproducible across platforms
and dependency versions. \texttt{supersedes} chains give it a history. The
tri-state verdict makes its epistemic status explicit. The migration of 400
legacy certificates, which validate 400/400 after migration, together with 400
fresh certificates and a byte-level determinism audit, is evidence that the
format can evolve without breaking the chain of trust.

\section{Limitations}\label{sec:limits}

The current limitations are listed explicitly, roughly in order of severity for
adopters.

\begin{enumerate}
\item \textbf{No learned policy yet.} The optimizer is a fixed-seed beam search.
  RL integration remains future work.

\item \textbf{The verifier shares its implementation with the certifier.} Both
  use PyZX for parsing, ZX reduction, and tensor computation, so a PyZX bug
  could in principle affect generation and verification in the same direction.
  The v0.5 mitigation is semantic: independent phase estimation, canonical
  notes, and a cross-engine fallback for the seven legacy exceptions. It is not
  a second engine. Cross-validation against QCEC or the MPO checker remains
  future work.

\item \textbf{Unitary-only certification.} Non-unitary records
  (\texttt{measure}, \texttt{if}, \texttt{reset}) are out of scope and fail soft
  to \texttt{inconclusive} since v0.4.1. 51 such instances were stripped
  from the benchmark at selection time.

\item \textbf{Sharp numeric threshold for continuous gate sets.} The V2 decision
  is a residual $\leq 2\times10^{-6}$ cut. Genuinely equivalent pairs observe
  ${\sim}\,2\times10^{-10}$ and genuinely different Clifford+T pairs
  $\gtrsim 10^{-3}$, but circuits with continuous gate sets could in principle
  land in the grey zone between. The tolerance is recorded in the certificate so
  consumers can re-decide under their own policy.

\item \textbf{Canonical-hash phase granularity.} Float rotation angles are
  canonicalized to multiples of $2^{-20}\pi$
  (\texttt{Fraction.limit\_denominator($2^{20}$)}). Angles differing by less
  than that quantum share a canonical hash, by design, and angles near a
  rounding boundary can hash differently across platforms.

\item \textbf{Numeric fallback limited to $\leq 6$ qubits.} Larger circuits rely
  entirely on the ZX path; ZX-inconclusive large circuits get no certificate.

\item \textbf{Up-to-global-phase equivalence only,} with the
  \texttt{not\_tracked} disclosure for circuits declaring an input global phase,
  which PyZX 0.10.4 drops.

\item \textbf{No DoS hardening.} The verifier and certifier place no limits on
  input size, gate count, or tensor dimension beyond the six-qubit numeric gate.
  A hostile certificate could request an expensive ZX reduction. A policy layer
  imposing size limits and timeouts belongs above the verifier and is not
  implemented.

\item \textbf{Seven legacy certificates verify via cross-engine fallback.} They
  verify through corroboration rather than exact reproduction of the declared
  engine (Section~\ref{sec:artifacts}). Their V2 re-issues are canonical.

\item \textbf{Strong baseline, modest headroom.} The PyZX baseline already
  captures most of the available improvement, a 21.4\% mean T-count reduction on
  the 82 non-Clifford circuits. The agent's strict advantage is real but narrow:
  9 circuits out of 100.

\item \textbf{Norm-divergence is structurally rare} (4/38 adversarial, 0/100
  standard), so the practical cost of choosing $T_G$ over $T_P$ is low on these
  distributions.

\item \textbf{The aggregate masks regression magnitude.} The clamped score
  reports a regressed objective as zero regardless of how large the regression
  is. The returned circuit never regresses (Proposition~1(iii)), but
  intermediate candidates are disclosed only through \texttt{metrics\_after} in
  their certificates.

\item \textbf{Hardware validation is preliminary.} A single-session experiment
  on IBM Heron~r2 (Section~\ref{sec:qpu}) confirms the ISA cost reduction and
  detects no statistically significant quality difference at 1024 shots per arm.
  The point estimates favour the optimized circuit, but the confidence interval
  also permits a small degradation. The interleaved multi-session protocol
  required for a significance claim is ongoing work. Noise-aware selection,
  layout, and drift remain outside the certificate's scope and are first-class
  objectives for future hardware-aware work.
\end{enumerate}

\section{Future Work}\label{sec:future}

Three directions are concrete.

\emph{Learned policies.} An RL policy, for example a GNN/PPO agent in the style
of RL-ZX \cite{rlzx}, or an OpenEvolve-style evolutionary search
\cite{openevolve}, behind the same hard gate, with the certificate corpus as the
reward-audit trail. The non-regression rank defines a clean shaped-reward
baseline.

\emph{Second-engine validation.} Emit QCEC- or MPO-compatible witness formats,
or re-verify AlchemQ certificates with an independent engine, removing the
shared-implementation caveat.

\emph{Protocol hardening.} A DoS policy layer above the verifier,
protocol-0.2.0 certificates inside the search loop itself since the loop
currently uses the 0.1.1 engine, canonicalization for non-unitary records, a
wider benchmark including QUBIKOS-style known-optimal instances, and the
interleaved multi-session hardware protocol of Section~\ref{sec:qpu}.

\section{Conclusion}\label{sec:conclusion}

AlchemQ v0.5 shows that a small untrusted optimizer and a certificate layer can
be composed into a system whose every output is independently verifiable,
byte-reproducible across platforms, and guaranteed by construction to be
equivalent to its input and no worse than the original on any objective.

The benchmark evidence is reported without flattery. The baseline is strong. The
agent's edge is narrow. The choice of t-norm is mostly irrelevant on standard
distributions. The system's most valuable outputs so far are the two defects its
certificates caught, one in PyZX and one in our own cross-platform
reproducibility. Both are results that benchmarks alone would not have produced,
because a benchmark reports a score while a certificate reports a failure with a
location.

The certificate protocol (0.2.0), the reference implementation, the full
artifact corpora, and the release manifest are released for reuse and
adversarial scrutiny.

\section*{Reproducibility}

Code, benchmark data, all certificate corpora, validation scripts, and the
release manifest accompany this paper in the repository cited below. The
2998-test suite passes in approximately 80\,s on both supported platforms
(Linux, Python 3.12, NumPy 2.2.5, PyZX 0.10.4, Qiskit 2.5.2; macOS, Python
3.14.4, NumPy 2.5.1, PyZX 0.10.4). The reference Linux gate run was
independently re-executed in a clean container prior to release; the original
gate logs are published under \pth{evidence/final-local-test/} and the
re-executed logs under \pth{release/}. Because the optimizer is proprietary and
its two test modules are therefore not published, the test count that executes
against the public repository is smaller than 2998; both numbers are stated
separately in the release manifest.

The benchmark is fully seeded (global seed 42). The gates of
Section~\ref{sec:bench} re-run via \pth{benchmark/gate\_g0.py} and
\pth{benchmark/gate\_g1.py}, whose recorded outputs are published alongside them.
The artifact-validation gates of Table~\ref{tab:validation} re-run via the
scripts under \pth{evidence/v04/scripts/}, \pth{evidence/v05/scripts/} and
\pth{evidence/certifier-v2/scripts/}. Aggregate tables are recomputed from
\pth{benchmark/full\_results.json} and \pth{results/results.json} under the score
semantics of \pth{alchemq/scoring.py}. The figures of Section~\ref{sec:qpu} are
recomputed from the raw measurement counts in
\pth{evidence/qpu/qpu-validation-ibm\_kingston.json} by \pth{qpu\_analyze.py}.

\section*{Code and Data Availability}

The certificate format specification and a standalone reference verifier
(protocols 0.1.1 and 0.2.0) are released under the Apache License 2.0 at
\url{https://github.com/a1damek1125/AlchemQ}, archived at
\href{https://doi.org/10.5281/zenodo.XXXXXXXX}{doi:10.5281/zenodo.XXXXXXXX}.
The benchmark corpus with its provenance manifest, the certificate corpora, the
closeout reports for both case studies, the gate and validation scripts, the
environment records, and the raw artifacts of the hardware experiment are
released in the same repository under CC BY 4.0. Every path referenced in the
Reproducibility section resolves there.

The verifier's numerical path depends only on the Python standard library and
NumPy. It re-parses the embedded QASM, rebuilds both unitaries, and reproduces
every hash, phase, and tolerance check. Reproducing the \texttt{gate-canon-v1}
canonical hashes of protocol 0.2.0 additionally uses PyZX, as an optional
version-pinned dependency, to parse the embedded QASM into its \texttt{Circuit}
IR. The verifier reproduces every published certificate and rejects all tested
tampering.

The optimization engine is proprietary to TriStiX S.L. This does not restrict
verification of any claim made here. A proof-carrying certificate is checkable
from the artifact alone, without access to or trust in the optimizer that
emitted it, so the open verifier is the complete artifact required to audit
every result in this paper. Binaries of the engine are available from the author
on request for academic evaluation.

\section*{Acknowledgments}

\paragraph{Use of generative AI tools.}
A large language model was used in the preparation of this manuscript, to draft
and revise portions of the prose and to assist with \LaTeX{} formatting and
table layout. It was not used to design the system, to run or select
experiments, or to produce any numerical value reported in this paper.

The research contributions are the author's own work: the certificate protocols
0.1.1 and 0.2.0, the equivalence and phase-estimation semantics, the reference
certifier and verifier, the benchmark construction and gating methodology, the
two certification case studies of Section~\ref{sec:eval}, and the design,
execution and analysis of the hardware experiment on the IBM Heron~r2 processor
\texttt{ibm\_kingston} (job \texttt{d9hjh7t0k0jc738i9phg}).

Every quantitative claim in this paper is recomputable from the released
artifacts rather than taken on the author's word. The aggregate tables of
Section~\ref{sec:eval} are recomputed from \pth{benchmark/full\_results.json} and
\pth{results/results.json} under the score semantics of \pth{alchemq/scoring.py}.
The figures of Section~\ref{sec:qpu} are recomputed from the raw measurement
counts in \pth{evidence/qpu/qpu-validation-ibm\_kingston.json} by
\pth{qpu\_analyze.py}. Readers are encouraged to re-derive them.

All LLM-assisted text was reviewed and edited by the author, and every reference
was checked against its source. The author takes full intellectual
responsibility for the entire content of this paper, irrespective of how any
portion of the text was generated. No generative AI system is listed as an
author.

\end{document}